\documentclass[12pt]{article}

\usepackage{amsthm,amsmath,amssymb,bbm,bm}
\usepackage{natbib}
\usepackage{multirow}
\usepackage[pdftex]{graphicx}
\usepackage{subfigure}
\usepackage{wrapfig}
\usepackage{tikz}
\usepackage{array}
\usepackage{url}

\usepackage{algorithmic}
\usepackage[linesnumbered, ruled, vlined]{algorithm2e}

\usepackage{mathrsfs}
\usepackage{dsfont}
\usepackage{titling}

\usepackage{relsize}
\usepackage{rotating}
\usepackage{enumitem}
\usepackage{float}
\usepackage{xcolor}
\usepackage{setspace}

{\textwidth=6.5in}

\newtheorem{example}{Example}
\newtheorem{definition}{Definition}[section]

\newtheorem{theorem}{Theorem}[section]
\newtheorem{lemma}{Lemma}[section] 
\newtheorem{proposition}{Proposition}[section]

\newcommand{\IMs}{IMs\xspace}
\newcommand{\IM}{IM\xspace}
\newcommand{\CC}{CC\xspace}
\newcommand{\CCs}{CCs\xspace}
\newcommand{\CD}{CD\xspace}
\newcommand{\CDs}{CDs\xspace}
\newcommand{\DS}{DS\xspace}
\newcommand{\GFD}{GFD\xspace}
\newcommand{\GFI}{GFI\xspace}

\newcommand{\calS}{\mathcal{S}}
\newcommand{\bel}{\operatorname{bel}_y}
\newcommand{\pl}{\operatorname{pl}_y}
\newcommand{\fid}{\operatorname{fid}_y}
\newcommand{\cc}{\operatorname{cc}_y}
\newcommand{\manifold}[1]{\mathcal M_{y,#1}}

\begin{document}


\title{\vspace{-3cm}  {\Large Demystifying Inferential Models: A Fiducial Perspective}}

\author{
Yifan Cui\thanks{Department of Statistics and Data Science, National University of Singapore},
Jan Hannig\thanks{Department of Statistics and Operations Research, University of North Carolina at Chapel Hill}}
\date{}

\maketitle
\vspace{-2.2cm}
\begin{abstract}
Inferential models have recently gained in popularity for valid uncertainty quantification. In this paper, we investigate inferential models by exploring relationships between inferential models, fiducial inference, and confidence curves. 
In short, we argue that from a certain point of view, inferential models can be viewed as fiducial distribution based confidence curves.  We show that all probabilistic uncertainty quantification of inferential models is based on a collection of sets we name principle sets and principle assertions.
\end{abstract}

\noindent {\bf keywords}
Confidence curve, Confidence distribution, Dempster-Shafer theory, Fiducial inference, Inferential model

\section{Introduction}

Inferential models (\IMs) \citep{martin2015inferential} 
 are one of the great statistical innovations of the 2010s. \IMs brought a thoroughly novel idea into the foundations of statistics by formalizing a way to assign epistemic probabilities to events that have guaranteed frequentist interpretation, called validity \citep{MartinLiu2013a,MartinLiu2013b,MartinLiu2013c,martin2015inferential,ryan2017IM}. 
When studying \IMs one cannot but notice that \IMs share many similarities with mathematical mechanics of fiducial inference.

R.A.~Fisher introduced the idea of fiducial probability \citep{Fisher1930}
as a potential replacement of the Bayesian posterior distribution. Although he discussed fiducial inference in several subsequent papers, there appears to be no rigorous universally  accepted definition of a fiducial distribution for a vector parameter. 
The basic idea of the fiducial argument is switching the role of data and parameters to introduce influentially meaningful distribution on the parameter space that summarizes our knowledge about the unknown parameter without introducing any prior information.

There are various related formalization of this idea such as Dempster-Shafer (DS) theory \citep{Dempster2008,  EdlefsenLiuDempster2009, shafer1976mathematical}, and generalized fiducial inference (GFI) \citep{Hannig2009,hannig2016generalized,Cui2019}
that has been applied to many modern statistical problems
 \citep{CisewskiHannig2012,WandlerHannig2012b, lai2015,hannig2016generalized,liu2016generalized,liu2017generalized,williams2019,williams2019b,Cui2019,Neupert2019,cui2020fiducial,cui2021unified}.
Even objective Bayesian inference, which aims at finding non-subjective model based priors can be seen as addressing the same basic question. Examples of recent breakthroughs related to reference prior and model selection are  \cite{BayarriEtAl2012, BergerBernardoSun2009, BergerBernardoSun2012}. There are many more references that interested readers can find in the review article \cite{hannig2016generalized}.

Another important idea in statistical foundations that received considerable interest in the past decade is confidence distribution (CD), that some viewed as ``the Neymanian interpretation of Fisher's fiducial distributions'' \citep{schweder2016confidence}.
\CD refers to a data-dependent distribution function that can represent confidence intervals (regions) of all levels for a parameter of interest \citep{XieSingh2013, schweder2016confidence}.
The initial idea can be traced back to early 20th century \citep{neymann1941,cox1958}, but \CDs have not received much attentions until the recent surge of interest in the research of \CD and its applications \citep{Efron1998,SchwederHjort2002,SchwederHjort2003,schweder2016confidence,XieSinghStrawderman2011,SinghXieStrawderman2005,singhxiestrawderman2007,XieSingh2013,luo2021leveraging,lawless2005,tian2011,yang2016,liu2014,Liu2015MultivariateMO,cui2021confidence}.
Heuristically speaking \CD is a function of both the parameter and the sample which satisfies two conditions. The first condition basically states that for any fixed sample, a
\CD must be a distribution function on the parameter space. The second condition essentially places a restriction to this sample-dependent distribution function such that the corresponding inference has desired frequentist properties.

A confidence curve (CC) was introduced by Birnbaum \citep{Birnbaum1961}. It was originally viewed as a useful graphical tool for visualizing CDs by converting, the distribution function $H_y(\theta)$ to the 
\begin{equation}\label{eq:initialCC}
    \cc(\theta)= 2|H_n(\theta)-0.5|.
\end{equation}
On a plot of a \CC defined in \eqref{eq:initialCC}, a line across the $y$-axis of the significance level $\alpha$, for any $0 < \alpha < 1$, intersects with the confidence curve at two points, that correspond to end points of the $\alpha$ level, equal tailed, two-sided confidence interval for $\theta$. In addition, the minimum of a confidence curve is the median of the \CD which can serve as a point estimator.
This idea has been further generalized beyond the equal-tailed univariate confidence intervals by Schweder and Hjort \citep{schweder2016confidence}.
They require that all $\{\theta : \cc(\theta)\leq \alpha\}$
form $\alpha$ confidence set, but these sets do not have to be intervals. The main advantage that \CC has over \CD is that it indicates the shape of the confidence set for each possible data $y$. 

The rest of the paper is organized as follows. 
In Section~\ref{s:IMdef}, we explain mathematical definition of \IMs, present a simple example, and discuss some basic properties of \IMs.
In Section~\ref{s:fiducialIM}, we explore relationships between \GFI and \IM. In particular, we investigate sets for which fiducial probability and \IM belief coincide. 
The concept of \CCs does not necessarily tell us how to obtain them. In Section~\ref{s:ccIM}, we show that \IMs can be viewed as a tool for obtaining \CCs. However they are not the only such tool, e.g., higher order likelihood inference
\citep{FraserFraserStaicu2010,
Fraser2004,
Fraser2011,
FraserNaderi2008,
Fraser:2005tc,
FraserReidMarrasYi2010}. 
Section~\ref{s:discussion} concludes by discussing our takes from the theorems proved in this paper. In particular we 
claim that \IMs can be, in some sense, viewed as fiducial distribution based confidence curves.

While we are using \GFI terminology in this paper, the same mathematical results can be formulated using the \DS theory. In particular, the  statements using fiducial probability can be replaced with the \DS belief function, which is different from the \IM belief function.

\section{Inferential models revisited}
\label{s:IMdef}
There are several definitions of \IMs that differ in details. In this section we present what we consider to be the most common version in the \IM literature. The definition uses three steps: association, prediction and combination.

{\em Association-step:}
In this step, one associates data, parameters and auxiliary random variables using a deterministic equation
\begin{equation}\label{eq:association}
    a(Y,\theta,U)=0.
\end{equation}
More precisely, this means that for all $\theta\in\Theta$ there exist random variables $U, Y$ defined on the same probability space satisfying \eqref{eq:association}, where
the auxiliary variable $U$ has a known distribution $f_U$ free of unknown parameters, for example $U(0,1)$, and $Y$ has density $f(y|\theta)$ with respect to some dominating measure. The function  $a(y,\theta,u)$ is assumed measurable. One of the consequences of the association is that for any observed $y$ there is at least one $u_{y,\theta}$ satisfying $a(y,\theta,u_{y,\theta})=0$. 

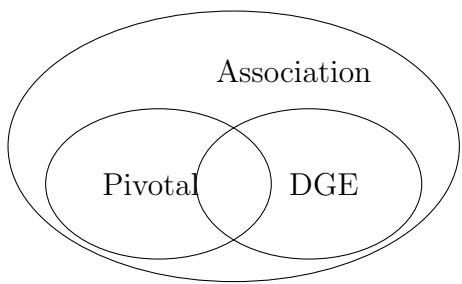
\begin{figure}[ht]
\centering
\begin{tikzpicture}
[xshift=7cm,yshift=-0.5cm]
      \draw (-1,0) ellipse (1.5cm and 1cm);
      \draw (1,0) ellipse (1.5cm and 1cm);
      \draw (0,0.5) ellipse (3cm and 1.8cm);
      \node (P) [xshift=-1.1cm,yshift=0cm]{Pivotal};
      \node (D) [xshift=1.2cm,yshift=0cm]{DGE};
      \node (A) [xshift=.8cm,yshift=1.5cm]{Association};
\end{tikzpicture}
\caption{A visualization of relationships between association, data generating and pivotal equations for a statistical model.}
\label{fig:1}
\end{figure}
A well-known example of association is the pivotal equation:
$T(Y,\theta)=U$, where if $Y$ follows $f(y|\theta)$ then the random variable $U$ has distribution free of unknown parameters. Another is the data generating equation (DGE) $Y=G(U,\theta)$. See Figure~\ref{fig:1} for a visualization of relationships between association, data generating and pivotal equations.

The association is then used to obtain a collection of sets $\Theta_y(u)$ of candidate parameter values, where $\Theta_y(u)\equiv \{\theta:a(y,\theta,u)=0\}$.

{\em Predict-step:}  Denote by $\theta_0$ the true parameter used in generating the observed data $y$. Denote by $u_{y,\theta_0}$ a value of $u$ that is associated by \eqref{eq:association} with the observed data $y$ and true parameter $\theta_0$. The goal of this step is to predict the value of $u_{y,\theta_0}$. This is done using a predictive random set {$\calS: \Omega \rightarrow 2^{R(U)}$, where $R(U)$ is the range of $U$. For example, if $U$ is $U(0,1)$, then $R(U)=[0,1]$ and one possible random set is $\mathcal S=[0,U^*]$, where $U^*$ is another $U(0,1)$ random variable.}

A good choice of $\mathcal S$ is crucial to the interpretation and properties of the \IM.  The following assumption is typically made on $\gamma(u) =  P(u \in \calS)$ \citep{liu2020inferential}: 
\begin{equation}\label{eq:Seq}
P(\gamma(U^\star)\leq\alpha)\leq\alpha
\mbox{ for all $\alpha\in(0,1),$}
\end{equation}
where $U^\star$ follows the same distribution as $U$ in \eqref{eq:association}. \footnote{Comment on notation: Throughout this paper we will denote by $U^\star$ a random variable that has the same distribution as the random variable used in the association step, often $N(0,1)$ or $U(0,1)$. On the other hand we will use $U^*$ to exclusively denote $U(0,1)$ random variable.}

{\em Combine-step:} \IM inference is based on a random set $\Theta_y(\calS) = \bigcup_{u\in \calS} \Theta_y(u)$.  Notice that $\Theta_y(\calS)$ is obtained using a process that is similar to de-pivoting in classical statistics \citep{CasellaBerger2002}.
Evidence in favor and against any assertion $A$ is based on a relationship between the assertion $A$ and the random set $\Theta_y(\calS)$.

There are a number of relevant summaries of the distribution of $\theta_y(\calS)$. In particular, the lower probability known as belief function is defined as
\begin{equation}
    \label{eq:belief}
\bel(A)\equiv P(\Theta_y(\calS) \subset A|\Theta_y(\calS) \neq \emptyset);
\end{equation}
the upper probability known as plausibility function is defined as: 
\[\pl(A)\equiv 1-\bel(A^c) = P(\Theta_y(\calS)\cap A \neq\emptyset|\Theta_y(\calS) \neq \emptyset). \]
If the set $\{\Theta_y(\calS) \subset A\}$ in \eqref{eq:belief} is not measurable, regular outer measure is used. 

In particular, $\bel(A)$ could be viewed as a measure of evidence for assertion $A$ while $1-\pl(A)$ as a measure of evidence against assertion $A$. The gap between $\bel(A)$ and $\pl(A)$ is an important feature of \IM. However, the following proposition shows that this gap is perhaps too large. 
\begin{proposition}\label{prop:gap}
If any two independent copies of the random set $\mathcal S$ have non-empty intersections with probability 1, then $\bel(A)>0$ implies $\pl(A)=1$. 
\end{proposition}
\begin{proof}
  Consider two independent copies of the random set $\calS_1$ and $\calS_2$. Recall $\bel(A)=P(\Theta_y(\calS_1)\subset A)>0$. If $\pl(A)<1$, then $P(\Theta_y(\calS_2)\subset A^\complement)>0$ which is a contradiction with the assumption  $P(\calS_1\cap\calS_2=\emptyset)=0$.
\end{proof}
 The implication of this proposition is that at most one of the probability bounds ($\bel(A),\pl(A)$) provides any useful information.
Notice that any nested random set trivially satisfies the assumption of non-empty intersection.

One of the key proprieties of \IMs is {\em validity} \citep{martin2015inferential}. While the following lemma is  well-known, we present its proof because it showcases an important technique used repeatedly in this manuscript.
\begin{lemma}
Assuming \eqref{eq:Seq}, then for any $A$,
\begin{equation}\label{eq:validity}
    \sup_{\theta\notin A}\bar P_\theta(\operatorname{bel}_Y(A)\geq 1-\alpha)\leq\alpha,
\end{equation}
where $\bar P_\theta$ is the regular outer measure associated with the likelihood of $Y$.
\end{lemma}
\begin{proof}
 Fix $\theta\notin A$ and because $\emptyset\subset A$ we have
 \[\bel(A)
 =P(\Theta_y(\calS) \subset A|\Theta_y(\calS) \neq \emptyset)\leq P(\Theta_y(\mathcal S)\subset A)\leq 1- \gamma(u_{y,\theta}),\]
 where $u_{y,\theta}$ is the $u$ value associated with $y$ and $\theta$. 
 Consequently, using the associated $Y$ and $U$
 \[
 \bar P_\theta(\operatorname{bel}_Y(A)\geq 1-\alpha)
 \leq P(1-\gamma(U)\geq 1-\alpha)=
 P(\gamma(U)\leq \alpha)\leq \alpha.
 \]
 The statement follows by taking supremum over $A^\complement$.
\end{proof}

\begin{example}
\label{ex:ex1}
We shall investigate \IM inference for the normal location model, i.e., $Y \sim N(\theta,1)$.

$A$-step:  The association between $Y$, $\theta$, and auxiliary variable $U$ can be expressed as $Y=\theta+\Phi^{-1}(U)$, where $\Phi$ is a distribution function of standard normal distribution. From here $\Theta_y(u)=\{y-\Phi^{-1}(u)\}$.

$P$-step:  One could predict the unobserved $u$ with a predictive random set $\calS = S_{U^*}$, where $U^*$ follows $U(0,1)$, and  $S_\alpha$ is defined by 
\footnote{There are other reasonable choices of $S_\alpha$ that lead to very different $\bel(A)$ values, e.g., $S_\alpha=(0.5\alpha,0.5+0.5\alpha)$, $S_\alpha=[0,\alpha]$, and $S_\alpha=[\alpha,1]$.} 
\begin{equation}
S_\alpha = \left\{ u \in (0,1): |u-0.5|<|\alpha-0.5|\right\},\quad \alpha \in(0,1).
\label{eq:S}
\end{equation}
It is easy to see that $\gamma(U^\star)\sim U(0,1)$ and \eqref{eq:Seq}
 is satisfied.

$C$-step: Considering the choice of predictive random set $S$ in Equation~\eqref{eq:S}, the random set $\Theta_y(\calS)$ is given by
\begin{multline*}
\Theta_y(\calS) 
=  \cup_{u\in S_{U^*}} \{y-\Phi^{-1}(u)\} \\
=  \Big(y-\Phi^{-1}( 0.5+|U^*-0.5|), y-\Phi^{-1}( 0.5-|U^*-0.5|)\Big).
\end{multline*}
The plausibility function is 
$
\pl(\theta) = 1- |2\Phi(y-\theta)-1|.
$ The \CC  based on the usual, two sided confidence intervals  \eqref{eq:initialCC} satisfies $\cc(\theta)=1-\pl(\theta)$.
Examples of $\cc(\theta)$ for this problem can be found in the left two panels of Figure~\ref{fig:2}.
\end{example}

We end this section by stating the following important lemma.
\begin{lemma}\label{lemma:nested}
Consider a random set $\mathcal S$ and corresponding $\gamma(u)$. 
The random set $\mathcal S'=S_{U^*}$, where $U^*$ is $U(0,1)$ and $S_\alpha=
\{u\, :\, \gamma(u)>1-\alpha\}$, shares the same $\gamma(u)$.
Moreover, if $P(\Theta_y(\mathcal S')\neq\emptyset)=1$, then $\bel(A)\leq\bel'(A)$ for all $A$, where $\bel'$ is the belief function associated with $\mathcal S'$.
\label{l:nested}
\end{lemma}
\begin{proof}
Compute:
\[\gamma'(u)=P(u\in\mathcal S')=P\left(1-U^*\leq \gamma(u)\right)=\gamma(u).\]
To prove the second part of the lemma, define ${\mathcal W}_y(A)=\{u \,:\, \Theta_y(u)\subset A\}$. We have 
\begin{multline*}
    \bel(A)\leq P(\mathcal S\subset {\mathcal W}_y(A))\leq 1-\sup\{\gamma(u): \ u\in {\mathcal W}_y(A)^\complement\}\\
    =P(\mathcal S'\subset {\mathcal W}_y(A))=\bel'(A),
\end{multline*}
where the last equality follows because of $P(\Theta_y(\mathcal S')\neq\emptyset)=1$.
\end{proof}

The meaning of this lemma is that in most situations only nested random sets $\mathcal S'$ defined in Lemma~\ref{l:nested} should be used. This is because given $\gamma(u)$ these random sets are the most efficient.

\section{Generalized fiducial distribution and inferential models}\label{s:fiducialIM}

Let us first quickly review the definition of generalized fiducial distribution (\GFD). Interested readers can find more details in \citep{hannig2016generalized}.
Given association \eqref{eq:association},
we define first the pseudo-solution of the association equation using the optimization problem:
\begin{equation}\label{eq:FIDopt}
 Q_{y}(u)=\arg\min_{\theta^{\star}} {\| a(y, \theta^{\star},u ) \|}.
\end{equation}
Typically, $\|\cdot\|$ is either $\ell_2$ or $\ell_\infty$ norm. If there is more than one solution to \eqref{eq:FIDopt}, one minimizer is selected based on some possibly random rule. Thus when there are multiple minimizers, there are multiple versions of fiducial distribution based on which one is selected.

Next, for each small $\epsilon>0$, define the random variable $\theta_\epsilon^\star=Q_\textbf{y}(U_\epsilon^{\star})$,
where $U_\epsilon^{\star}$ has the distribution of $U$ {\em truncated} to the set 
\begin{equation}\label{eq:truncate}
\manifold{\epsilon} = \{ U_\epsilon^{\star} : \|a(y, U_\epsilon^{\star}, \theta_\epsilon^\star ) \| = \| a(y, U_\epsilon^{\star}, Q_y(U_\epsilon^{\star}) ) \| \leq \epsilon \},
\end{equation}
i.e., having the density $f_U(u^\star) I_{ \manifold{\epsilon}}(u^\star)/(\int_{\manifold{\epsilon}} f_U(u)\,du),$ where $f_U$ is the density of $U$.
Then assuming that the random variable $\theta_\epsilon^\star$ converges in distribution as $\epsilon\to 0$, the GFD is defined as the limiting distribution of $\theta^\star=\lim_{\epsilon\to 0} \theta_\epsilon^\star$. The fiducial probability is then
$\fid(A)=P(\theta^\star\in A)$, where the probability is based on the distribution of $\theta^\star$ with the observed data $y$ taken as fixed. 
Clearly, if $P(U^\star\in \manifold{0}) >0$ then $\theta^\star=\theta^\star_0$, and 
\[\fid(A)=P(Q_y(U^\star)\in A\mid U^\star\in\manifold{0}),\]
where again $U^\star$ has the same distribution as $U$ in \eqref{eq:association}.

Notice that in the \GFD literature the association is usually assumed to be based on a DGE. While any association equation, can be used to conduct statistical inferences, only the data generating equation is guaranteed not to lose information. This loss of information can occur when there are multiple $y$ for each $\theta$ and $u$ solving the association equation, and the inference is based on insufficient statistics. 

The following theorems provide connections between \IMs and \GFI. They are established in terms of belief though similar statements can be derived for plausibility.

\begin{theorem}
Consider an association \eqref{eq:association}, assume $P(U^\star\in\manifold{0})=1$, and that the predictive random set satisfies \eqref{eq:Seq}. Then 
$\bel(A)\leq \fid(A)\leq \pl(A)$ for any measurable $A$.
\label{thm:fiducial1}
\end{theorem}

\begin{proof}
Define $\underline{\mathcal U}_y(A)=\{u\in \manifold{0}\,:\, \Theta_y(u)\subset A\}=\mathcal{W}_y(A)\cap\manifold{0}$, which is measurable because $A$ is measurable. Next, define $\overline{\mathcal U}_y(A)=\underline{\mathcal U}_y(A^\complement)^\complement$.
Note that because $P(U^\star\in\manifold{0})=1$, there is no conditioning and 
\[ P\left(U^\star\in \underline{\mathcal U}_y(A)\right)\leq \fid(A)\leq P\left(U^\star\in \overline{\mathcal U}_y(A)\right).\]
Moreover, $\Theta_y(\mathcal S)\subset A$ if and only if $\mathcal S\subset \underline{\mathcal U}_y(A)$ and therefore
\begin{equation*}
    \bel(A)=P(\mathcal S\subset \underline{\mathcal U}_y(A))\leq 1-\sup\{\gamma(u): \ u\in \underline{\mathcal U}_y(A)^\complement\}\equiv 1-\beta.
\end{equation*}
Finally, $\underline{\mathcal U}_y(A)\supset \{u:\gamma(u)>\beta\}$ and
\eqref{eq:Seq} imply
\[ P\left(U^\star\in \underline{\mathcal U}_y(A)\right)\geq P(\gamma(U^\star)>\beta)\geq 1-\beta.\]
The rest of the proof follows by the definition of $\pl(A)$.
\end{proof}

\begin{theorem}\label{thm:IMfiducialAttained}
Consider association \eqref{eq:association} with $U$ having a non-atomic distribution.
For any fixed measurable set $A$ and fixed data $y$, there exists a nested random set $\mathcal S$ satisfying \eqref{eq:Seq} and a version of fiducial distribution satisfying $\bel(A)=\fid(A)$. 
\end{theorem}

\begin{proof}
We will be using the same notation as in the proof of Theorem~\ref{thm:fiducial1}. 
If $P(U^\star\in\manifold{0})>0$, then
\[ P\left(U^\star\in \underline{\mathcal U}_y(A)\mid \manifold{0}\right)\leq \fid(A)\leq P\left(U^\star\in \overline{\mathcal U}_y(A)\mid \manifold{0}\right).\]
We can select the minimizer in \eqref{eq:FIDopt} so that $P\left(U^\star\in \underline{\mathcal U}_y(A)\mid \manifold{0}\right)= \fid(A)$. 
If  $P(U^\star\in\manifold{0})=0$ we can use any minimizer in \eqref{eq:FIDopt}.

Next select a collection of sets $S_\alpha,\ \alpha\in[0,1]$ satisfying all of the following: 
a) $S_\alpha\cap\manifold{0}\neq\emptyset$;  
b) $S_{\alpha_1}\subset S_{\alpha_2}$ whenever $\alpha_1<\alpha_2$;
c)  $P(U^\star\in S_\alpha)=\alpha$;
d) $S_{\fid(A)}\cap \manifold{0}=\underline{\mathcal U}_y(A)$; 
e) $S_\alpha\cap \underline{\mathcal U}_y(A)^\complement\neq\emptyset$ whenever $\alpha>\fid(A)$. Such collection always exists but it is not unique. 
The random set $\mathcal S=S_{U^*}$, where $U^*$ follows $U(0,1)$, satisfies the statements of the theorem. 
\end{proof}

Theorems~\ref{thm:fiducial1} and \ref{thm:IMfiducialAttained} show that fiducial probability is a natural bound for IM beliefs and plausibility. This bound can be achieved for any measurable set $A$ using some random set $\mathcal S$. Thus one could argue that IM is less efficient than fiducial probability, i.e., the beliefs are perhaps too small and plausibilities too high. However, this gap allows \IM to gain some favorable properties, such as avoiding false confidence and guaranteeing frequentist validity. However, as seen in Theorem~\ref{thm:IMfiducialAttained}, \IM can guarantee these good properties only if the random set is selected {\em before} seeing the data. Therefore data snooping should be avoided when using \IM. 

Next, we show that there are many sets where belief and fiducial probabilities agree. These sets will be important in establishing a connection between \IM and confidence curves.

\begin{theorem}\label{thm:IMbasis}
Assume that the random set  $\mathcal S$ is nested. For any data $y$ and any $\alpha\in[0,1]$, there exists a set $A_{\alpha,y}$, so that $\bel(A_{\alpha,y})\geq\alpha$, and $A_{\alpha,y}\subset A$ for all $A$ satisfying $\bel(A)\geq\alpha$, and for any $\alpha_1<\alpha_2$, we have $A_{\alpha_1,y}\subset A_{\alpha_2,y}$.

If additionally $P(U^\star\in\manifold{0})=1$ and in \eqref{eq:Seq} we have $\gamma(U^\star)\sim U(0,1)$, then there is a version of $\fid$ satisfying $\bel(A_{\alpha,y})=\fid(A_{\alpha,y})=\alpha$ for all $\alpha$ in the range of $\bel$.
\end{theorem}

\begin{proof}
Recall that the random set $\mathcal S$ is nested.
Denote the range of the random set by $R(\mathcal S)$ and define 
\begin{equation}\label{eq:principlenestedsets}
S_\alpha=\bigcup_{\beta<\alpha}\bigcap\{S\in R(\mathcal S)\,:\, P(\mathcal S\subset S)\geq \beta\}. 
\end{equation}
Because the random set is nested and by continuity of measure, $P(\mathcal S\subset S_\alpha)\geq\alpha$. By definition, for any $\alpha_1<\alpha_2$ we have $S_{\alpha_1}\subset S_{\alpha_2}$.

Define $A_{\alpha,y}=\Theta_y(S_\alpha)$. Clearly, $\bel(A_{\alpha,y})\geq\alpha$ and the sets $A_{\alpha,y}$ are nested. 
If $\bel(A)\geq\alpha$, then by definition $A_{\alpha,y}\subset A$. Moreover, if $\bel(A)=\alpha$ then $\bel(A_{\alpha,y})=\alpha$.

Under the additional assumptions and following the proof of Theorem~\ref{thm:fiducial1} we have
$\fid(A_{\alpha,y})=P(U^\star\in S_\alpha)=P(\gamma(U^\star)>1-\alpha)=\alpha.$ Here we remark that the chosen version of $\fid$ selects one of the minimizers in \eqref{eq:FIDopt} that belongs to $A_{\alpha',y}^\complement$ with the largest possible $\alpha'$.
\end{proof}

Notice that Theorem~\ref{thm:IMbasis} implies $\bel(A)=\sup\{\alpha : A_{\alpha,y}\subset A\}.$ Therefore we will call the collection of $A_{\alpha,y}$ the {\em principle assertions} of \IM as they carry all the information available in the \IM. In the next section we will see 
that principle assertions provide an important link to confidence curves. Similarly we will call $S_\alpha$ in \eqref{eq:principlenestedsets} the {\em principle nested sets}.

\section{Inferential models and confidence curves}
\label{s:ccIM}

Confidence curves could be viewed as functions whose level sets provide a collection of valid confidence intervals at all levels of confidence.
The formal definition is given below:
\begin{definition}
  A function $\cc(\theta)\,: \mathcal Y \times \Theta\to[0,1]$ is a \CC if for all $\alpha\in(0,1)$,
  $P_\theta(\operatorname{cc}_Y(\theta)< \alpha)=\alpha$.
Similarly $\cc(\theta)$ is a conservative \CC if 
  $P_\theta(\operatorname{cc}_Y(\theta)< \alpha)\geq\alpha$,
  for all $\alpha\in(0,1)$.
\end{definition}
In other words, the random confidence curve $\operatorname{cc}_Y(\theta)$ is uniformly distributed when $\theta$ is the true parameter.
Confidence curve generalizes the idea of confidence distribution \citep{XieSingh2013} to higher-dimensional parameters, where its contours provide a nested family of confidence regions indexed by degree of confidence \citep{schweder2016confidence}, i.e., each  $\{\theta\,:\, \cc(\theta)\leq\alpha\}$ is an $\alpha\cdot 100\%$ confidence set. 

\begin{example}
  We continue Example~\ref{ex:ex1}.
  Suppose that in addition to $Y \sim N(\mu_y, 1)$ we also have  $X \sim N(\mu_x, 1)$. We plot confidence curves for $\mu_x$ in the left panel, $\mu_y$ in the middle panel, and $\mu_x/\mu_y$ in the right panel of Figure~\ref{fig:2}. The confidence curve for $\mu_x/\mu_y$ is based on Fieller's confidence set \cite{Fieller1954}. Notice that the shape of the \CC indicates that based on $\alpha$ the confidence set for $\mu_x/\mu_y$ might be an interval, two disjoint intervals, or the whole real line.
\end{example}

\begin{figure}[ht]
 \centering
\includegraphics[width=0.32\textwidth]{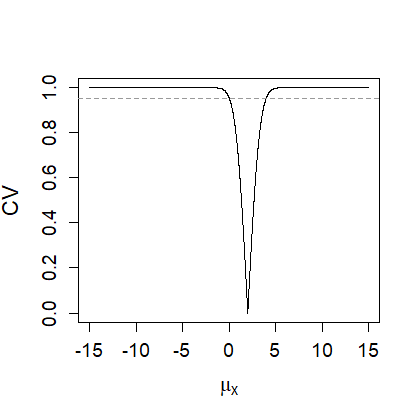}
 \includegraphics[width=0.32\textwidth]{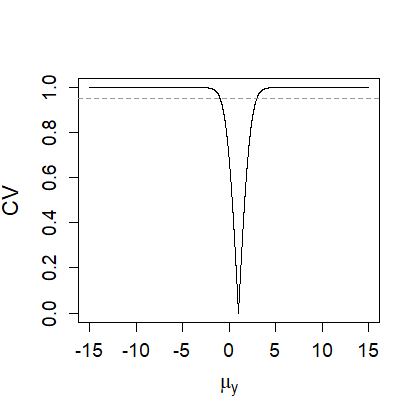}
 \includegraphics[width=0.32\textwidth]{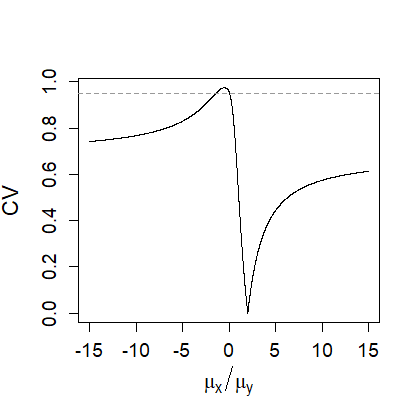}
\caption{Confidence curves for $\mu_x,\mu_y,\mu_x/\mu_y$ where $X\sim N(\mu_x,1)$, $Y\sim N(\mu_y,1)$ given that we observe $x=2,y=1$. The dotted lines indicate 95\% confidence sets.}
\label{fig:2}
\end{figure}

We shall investigate the relationship of \CCs to \IMs next.
 Let us first show that \CC can be used to define valid belief and plausibility functions without the use of any \IM model.
 \begin{lemma} Let $\cc(\theta)$ be a conservative \CC.
  The function  $\operatorname{cc-bel}_y(A)=\sup\{\alpha : \{\theta : \cc(\theta)\leq\alpha\}
  \subset A\}$ is a valid belief function.
 \end{lemma}
 \begin{proof}
   When $\operatorname{cc-bel}_y(A)\geq 1-\alpha$ then $\cc(\theta)\geq 1-\alpha$ for all $\theta\notin A$. Thus
   \[
   P_\theta(\operatorname{cc-bel}_Y(A)\geq 1-\alpha)
     \leq P_\theta(\operatorname{cc}_Y(\theta)\geq 1-\alpha)
     =1-P_\theta(\operatorname{cc}_Y(\theta) < \alpha)
     \leq \alpha
   \]
   for any $\theta\notin A$.
 \end{proof}
Notice that 
$\operatorname{cc-pl}_y(A)=1-\operatorname{cc-bel}_y(A^\complement)=\inf\{\cc(\theta) : \theta\in A\}.$
Therefore $\cc(\theta)$ is a possibility contour in the sense of \cite{zadeh1978fuzzy,dubois2012possibility,dubois2006possibility,augustin2014introduction}. 
Further connections between \IM and possibility theory has been explored in \cite{liu2020inferential}.

The previous lemma is not surprising in light of the next theorem which shows that any \CC can be formally viewed as an instance of a valid \IM. 
\begin{theorem}
{Given a confidence curve, there exists an \IM satisfying \eqref{eq:Seq} which provides inference equivalent to this confidence curve.}
\end{theorem}

\begin{proof}
 Denote by $\eta(u)$ the inverse of the distribution function of $\operatorname{cc}_Y(\theta)$. Because by definition of \CC this is sub-uniform, $\eta(u)\leq u$.
 Thus we can consider the association:
 \begin{equation*}
   a(y,\{\theta,\eta\},u)=\cc(\theta)-\eta(u).
 \end{equation*}
 If the \CC is exact, $\eta(u)=u$, otherwise $\eta(u)\leq u$ will be viewed as an unknown parameter.
 
 The random set is taken as $\mathcal S=[0,U^*]$, where $U^*\sim U(0,1)$.
Note that $S_\alpha=[0,\alpha)$, $\gamma(u)=1-u$ and \eqref{eq:Seq} is trivially satisfied.
 
 We will use the marginal \IM with \[\Theta_y(S)=\{\theta : 
a(y,\{\theta,\eta\},u)=0, \mbox{ for some $u\in S$ and $\eta$}\}.\]
 Thus
$\pl(\{\theta\})=P(\theta\in\Theta_y(\mathcal S))=P(\cc(\theta)\leq U^*)=1-\cc(\theta)$. Similarly,
the principle assertions $ A_{\alpha,Y}= \{\theta\,:\, \cc(\theta)\leq\alpha\}$ are exactly the same as the confidence sets implied by the \CC. Consequently, $\operatorname{cc-bel}_y(A)=\bel(A)$ for all $A$.
\end{proof}

Finally, we show that any \IM defines an associated \CC. It also shows that fiducial probability plays a special role in this association.
\begin{theorem}\label{thm:IM2CC}
 Assume \eqref{eq:Seq}. The sets $A_{\alpha,y}$ defined in Theorem~\ref{thm:IMbasis} are $\alpha\cdot 100\%$ confidence sets, and $\cc(\theta)=\inf\{\alpha\,:\,\theta\in A_{\alpha,y}\}$ is a conservative \CC. Additionally, if $P(U^\star\in\manifold{0})=1$ then the actual coverage of $A_{\alpha,y}$ is given by $\fid(A_{\alpha,y})$, where we use the version of $\fid$ introduced in the proof of Theorem~\ref{thm:IMbasis}. Consequently, $\cc'(\theta)=\inf\{\fid( A_{\alpha,y})\,:\,\theta\in A_{\alpha,y}\}\geq\cc(\theta)$ is an exact \CC. 
\end{theorem}

\begin{proof}
  We will be using the notation introduced in Section~\ref{s:fiducialIM}.
  Let $\theta$ and $U$ be the parameter and ancillary random variable associated with data $Y$ by \eqref{eq:association}. Then, just as in the proof of Theorem~\ref{thm:fiducial1},
  \begin{equation}\label{eq:IM2CC}
   P_\theta(\theta\in A_{\alpha,Y})=P(U\in S_\alpha)\geq P(\gamma(U)>1-\alpha)\geq\alpha.
  \end{equation}
 Additionally, if $P(U^\star\in\manifold{0})=1$ then  $\fid(A_{\alpha,y})=P(U^\star\in S_\alpha)$, which does not depend on $y$, and as seen in the first equality in \eqref{eq:IM2CC} is the coverage of $A_{\alpha,y}$.
 
 The statements about \CCs follow by definition.
\end{proof}

\section{Discussion} \label{s:discussion}

The results in Section~\ref{s:ccIM} show that the key concept of validity is innate to \CCs. In other words, \IMs are valid when they produce valid \CCs. However,  the property of \CC does not come with a recipe on how to produce it. The big advantage of \IMs is that its three-step recipe provides a systematic way to produce \CCs.

As seen in Theorem~\ref{thm:IM2CC} principle assertions $A_{\alpha,y}$ provide the main link between \IM and \CC. It may be of interest to point out that $A_{\alpha,y}$ is a confidence set that is obtained by de-pivoting of the principle nested set $S_\alpha$ using the association \eqref{eq:association}. Moreover, its coverage is given by the fiducial probability $\fid(A)$. 

In other words, \IM answers an old question, when are (fiducial) credible intervals confidence intervals, i.e., what sets of fiducial probability $\alpha$ are $\alpha\cdot 100\%$ confidence sets? The answer given in Theorem~\ref{thm:IMfiducialAttained} is that technically it could be any measurable set as long as for all the unobserved data $y'$ we would use the corresponding principle assertion $A_{\alpha,y'}$ de-pivoted from the principle nested set $S_\alpha$ from the proof of Theorem~\ref{thm:IMfiducialAttained}. This is of course unrealistic and only set $A_{\alpha,y}$ obtained by de-pivoting a sensible $S_\alpha$ chosen prior to seeing the data will be a set of fiducial probability $\alpha$ that is also $\alpha\cdot 100\%$ confidence sets. 

While most of the results linking \GFD and \IM are proved under the assumption $P(U^\star\in\manifold{0})=1$, the same results remain true if conditioning on $\{U^\star\in\manifold{0}\}$ is effectively conditioning on ancillary statistic \citep{MartinLiu2013b}, e.g., maximal invariant sigma algebra of a group model. If the conditioning in \GFD is not conditioning on ancillary statistic, the relationship between \IM and \GFD is complicated and needs to be studied further. 

In conclusion, \IM can be viewed as fiducial distribution based confidence curves. They provide a powerful argument for anyone seeking fiducial or objective Bayes distributions on parameter space to consider making calculations on the auxiliary $U$ and pay attention to how are the credible sets linked across all potential unobserved data $y'$.

\section{Acknowledgements}
{Jan Hannig's research was supported in part by the National Science Foundation under Grant No. DMS-1916115 and 2113404.}

\bibliographystyle{asa}
\bibliography{IMFiducial}

\end{document}